\documentclass[12pt]{article}
\usepackage[english]{babel}
\usepackage{amsfonts,amsmath,amsthm,graphicx,a4wide,color}
\usepackage[colorlinks=true]{hyperref}

\newtheorem{lem}{Lemma}
\newtheorem{thm}{Theorem}

\newcommand{\dotex}{\frac{d}{dt}}
\newcommand{\tr}[1]{\text{Tr}\left(#1\right)}
\newcommand{\ra}[1]{\text{Rank}\left(#1\right)}

\newcommand{\NN}{{\mathbb N}}
\newcommand{\RR}{{\mathbb R}}
\newcommand{\CC}{{\mathbb C}}

\newcommand{\bra}[1]{\langle #1 |}
\newcommand{\ket}[1]{| #1 \rangle}
\newcommand{\q}[1]{| #1 \rangle}
\newcommand{\bket}[1]{\left\langle #1 \right\rangle}

\newcommand{\pos}[1]{#1^{p}}
\renewcommand{\neg}[1]{#1^{n}}
\newcommand{\Lnorm}[1]{\left\|#1\right\|_{L}}

\newcommand{\bH}{\boldsymbol{H}}
\newcommand{\bI}{\boldsymbol{I}}
\newcommand{\ba}{\boldsymbol{a}}

\newcommand{\bB}{\boldsymbol{B}}

\newcommand{\bL}{\boldsymbol{L}}
\newcommand{\bS}{\boldsymbol{S}}

\newcommand{\bzq}{\boldsymbol{Q}}
\newcommand{\bBl}{\boldsymbol{B}}
\newcommand{\bM}{\boldsymbol{M}}
\newcommand{\bN}{\boldsymbol{N}}

\newcommand{\bR}{\boldsymbol{R}}

\newcommand{\Ktr}{\mathcal{K}^{1}(\mathcal{H})}

\newcommand{\KL}{\mathcal{K}_{L}(\mathcal{H})}
\newcommand{\Kf}{\mathcal{K}^f(\mathcal{H})}

\newcommand{\fA}{\mathfrak{A}}
\newcommand{\fB}{\mathfrak{B}}
\newcommand{\fC}{\mathfrak{C}}
\newcommand{\fK}{\mathfrak{K}_\infty}

\newcommand{\cH}{\mathcal{H}}
\newcommand{\cD}{\mathcal{D}}

\newcommand{\as}{{}}

\begin{document}

\title{Well-posedness and convergence of the Lindblad master equation for a quantum harmonic oscillator with multi-photon drive and damping
\thanks{This work was partially supported by the Projet Blanc ANR-2011-BS01-017-01 EMAQS and by the Belgian IAP programme DYSCO. The authors of this research are with the QUANTIC joint project team between INRIA,  ENS Paris,   Mines-ParisTech, CNRS and  Universit\'e Pierre et Marie Curie (Paris 6). } }
\author{R\'{e}mi Azouit\thanks{Centre Automatique et Syst\`{e}mes, Mines-ParisTech, PSL Research University. 60 Bd Saint-Michel, 75006 Paris, France. }
 \and Alain Sarlette\thanks{INRIA Paris, France; and Ghent University / SYSTeMS, Technologiepark 914, 9052 Zwijnaarde, Belgium. }
 \and Pierre Rouchon\thanks{Centre Automatique et Syst\`{e}mes, Mines-ParisTech, PSL Research University. 60 Bd Saint-Michel, 75006 Paris, France.}}

\date{\today}
\maketitle

\begin{abstract}
We {\as consider the model of a quantum harmonic oscillator governed by a Lindblad master equation where the typical drive and loss channels are multi-photon processes instead of single-photon ones; this implies a dissipation operator of order $2k$ with integer $k>1$ for a $k$-photon process. We prove that the corresponding PDE  makes the state converge, for large time, to an invariant subspace spanned by a set of $k$ selected basis vectors; the latter physically correspond to so-called coherent states with the same amplitude and uniformly distributed phases.} We also show that this convergence features a finite set of bounded invariant {\as functionals of the state (physical observables)}, such that the final state in the {\as invariant} subspace can be directly predicted from the initial state. The proof includes the full arguments towards the well-posedness of the corresponding dynamics in  proper Banach spaces of Hermitian trace-class operators equipped with adapted nuclear norms. It relies on the Hille-Yosida theorem and Lyapunov convergence analysis.
\end{abstract}

\paragraph{Mathematics Subject Classification:} 37L99, 47B44, 81Q93, 81S22, 81V10.
\paragraph{keywords:} infinite-dimensional dissipative dynamical systems, Lyapunov functions and stability, accretive operators, Lindblad master equation,  decoherence, quantum control, quantum electrodynamics and circuits.


\section{Introduction: open quantum systems}

The state of an isolated quantum system is notably described by a wave-function $\q{\psi}$ on a separable Hilbert space $\cH$. The evolution of $\q{\psi(t)}$ is described by the Schr\"odinger equation
$$\tfrac{d}{dt} \ket{\psi} = \frac{-i}{\hbar}\, H\, \ket{\psi} $$
where the Hamiltonian $H$ is a Hermitian operator on $\cH$. This equation implies a unitary evolution in $\cH$, i.e.~denoting $\bket{\psi | \phi}$ the scalar product between $\psi,\phi \in \cH$ and $\| \ket{\psi} \|_{\cH} = \sqrt{\bket{\psi | \psi}}$ the associated norm, we start with a normalized wave-function $\| \q{\psi(0)}\|_{\cH}=1$ and we are ensured to keep $\| \q{\psi(t)}\|_{\cH}=1$ for all $t$. Thus under such so-called Hamiltonian evolution, the Hilbert-Schmidt distance $\| \q{\psi(t)}-\q{\phi(t)} \|_{\cH}$ between two different initial states remains invariant in time. This is not suitable for control purposes, where we want to drive an initially unknown state towards a target value.

Doing the latter thus requires to consider open quantum systems, i.e.~systems interacting with their environment (see~\cite[Chapter 4]{HarocheBook} for a  recent  physical  introduction to decoherence  and  \cite{QOtheory,Davies76qbook} for  more  formal and mathematical  presentations). The most drastic interaction of this type is the famous projective measurement described by Von Neumann, described by an Hermitian  operator $Q:\cH \mapsto \cH$ with spectral decomposition and where the state gets projected onto the eigenspace of $Q$ corresponding to the measurement result. At the other end of the interaction spectrum, the target quantum system can be in weak interaction with an unobserved large environment which rapidly forgets its past state. We can then only describe the expected evolution of the target system, which under appropriate assumptions ensuring essentially the Markovian character of the evolution, follows a so-called Lindblad master equation \cite{QOtheory},\cite[Chapter 8.4]{NielsenChuang},\cite[Chapter 4]{HarocheBook}
$$
\tfrac{d}{dt} \rho = \tfrac{-i}{\hbar}\, [\bH,\rho] + \sum_j \boldsymbol{L}_j\rho \boldsymbol{L}_j^\dagger - \tfrac{1}{2} \boldsymbol{L}_j^\dagger \boldsymbol{L}_j \rho - \tfrac{1}{2} \rho \boldsymbol{L}_j^\dagger \boldsymbol{L}_j \, .
$$
Here the $\boldsymbol{L}_j$ can be a priori arbitrary operators on $\cH$, $\boldsymbol{L}_j^\dagger$ is the adjoint (i.e., Hermitian conjugate)  of $\boldsymbol{L}_j$ and $\rho$ is a density operator, i.e.~a nonnegative Hermitian   (i.e.,  self-adjoint) operator on $\cH$ with $\tr{\rho}=1$. When $\rho = \ket{\psi}\bra{\psi}$ i.e.~$\ra{\rho}=1$, and $\boldsymbol{L}_j=0$ $\forall j$, we recover the Schr\"odinger equation.

Lindblad type evolution, unlike Schr\"odinger type evolution, can make the state $\rho$ converge asymptotically towards a subspace or a unique state \cite[Chapter 8.4]{NielsenChuang},\cite{Tic1,Tic2}. From a control engineering viewpoint, it is hence tempting to design a Lindblad type system such that it stabilizes some target states thanks to the interaction of the system with an environment, much like the Watt governor does for the steam engine \cite{WattGov}. In quantum control this is called \emph{reservoir engineering} and it has been successfully applied to stabilize quantum states of interest without requiring explicit measurement feedback, see e.g.~\cite{ResEng1,ResEng2,ResEng3,ResEng0a,ResEng0b,ResEng0c,ResEng0d,legthas2015science}. Besides the technological advantage of working without sensor feedback, this also enables deterministic stabilization, since the Lindblad equation describes a deterministic evolution whereas individual quantum measurements follow a stochastic process.

Several reservoir engineering proposals, with potential benefits for quantum technology applications, consider infinite-dimensional Hilbert spaces like the ones describing harmonic oscillation of structures (phonons) or of an electromagnetic field at a given frequency \cite{ResEng1,ResEng2,ResEng3}. The latter is subject to intense development in cavity- or circuit-Quantum ElectroDynamics experiments and truly quantum states of the electromagnetic field, with no classical equivalent, have been stabilized experimentally. In particular, our collaborators have recently proposed a scheme to stabilize ``$k$-legged Schr\"odinger cat states'', i.e.~a quantum superposition of $k$ electromagnetic field states with the same amplitude but $k$ different phases \cite{legthas2015science}. From physical arguments and an invariance analysis, they argue that their engineered reservoir stabilizes a subspace of all possible electromagnetic field states, spanned by $k$ such states. Such so-called protected subspace or decoherence-free subspace could then be used for encoding quantum information in quantum memories or quantum telecommunication applications \cite{MirrahimiCatComp2014}.
The purpose of the present paper is to rigorously establish the convergence properties of this engineered reservoir in the infinite-dimensional framework.\\

From a mathematical viewpoint, rigorous analysis of Lindblad type dynamics with $\cH$ infinite-dimensional is nontrivial, as even the appropriate space for solutions $\rho(t)$ has to be specified. Physicists usually rely on a practical combination of physical arguments, invariance analysis, and (at best) convergence analysis of a finite-dimensional truncation to convince themselves of the soundness of a proposed Lindblad evolution, before confirming it by experimental implementation. However it is known that this holds traps, as phenomena like loss of probability mass to infinity can appear in some theoretical models. The presence of unbounded operators $\boldsymbol{L}_j$ in the Lindblad equation requires particular care \cite{Davie1977RoMP,Davies76qbook}, and with the model of \cite{MirrahimiCatComp2014,legthas2015science} we are precisely in this case. Similar questions and mathematical issues  relative to   well-posedness and long-time behavior of dissipative infinite dimensional quantum systems  have been addressed, for example,  in~\cite{AArnold} where  the  density operator $\rho$ governed by a Lindblad master equation   is replaced  by  the Wigner pseudo-probability distribution governed by a $3D$  integro-partial differential equation describing the  evolution of an electron ensemble connected to an idealized heat bath under the single-particle Hartree approximation.

The contribution of the present paper is precisely to provide a rigorous analysis for the Lindblad dynamics proposed in \cite{MirrahimiCatComp2014,legthas2015science}. {\as We introduce appropriate Banach spaces for the well-posedness} (Section \ref{sec:WP}, theorem~\ref{thm:Cauchy}); we prove the asymptotic convergence, as desired, of any initial physical state within this solution space towards a decoherence-free subspace of dimension $k$, spanned by the $k$-legged Schr\"odinger cat states (Section \ref{sec:AC}, theorem~\ref{thm:convergence}); and we characterize the limit point attained by any initial state by establishing the existence of $k^2$ physical quantities, attached to  $k^2$ linear bounded operators on $\cH$,  which remain invariant under the Lindblad evolution (Section \ref{sec:AC}, theorem~\ref{thm:invariants}).
The essential ingredients for our approach, partially inspired from~\cite{Davie1977RoMP}, are:
\begin{itemize}
\item the positivity of the density operators $\rho$ and associated trace-class operators, or their decomposition into positive and negative parts;
\item a particular commutation property for our specific system, namely that the Lindblad operator $\bL$ describing such  $k$-photon exchanges is such that  $[\bL,\bL^\dagger]$ is positive definite {\as (see~\eqref{eq:M} and  \eqref{eq:trLPIL} in the proof of Lemma 3);}
\item building an adapted  nuclear norm from the Lindblad operator $\bL$ {\as (see lemma~\ref{lem:NormL})};
\item exploiting this norm via the Hille-Yosida theorem for well-posedness, and via the Lyapunov function $\tr{\bL \rho\bL^\dag}$ for  convergence analysis;
\item  a density and duality argument to prove the existence of  $k^2$ invariant bounded operators, followed by particular Fourier-transform-like insight to explicitly identify some of them.
\end{itemize}
Preliminary results, for the case $k=2$ and without investigating the well-posedness, are available in \cite{ourCDC}.


\section{Lindblad model for harmonic oscillator with $k$-photon exchange}

The harmonic oscillator is the most basic model of a quantum system on an infinite-dimensional Hilbert space \cite{HarocheBook}. Starting with the canonical orthogonal basis $\{\ket{n}\}_{n\in\NN}$ (Fock basis) corresponding to the photon-number states, i.e.~$\q{\psi} = \q{n}$ corresponds to an oscillator state with exactly $n$ quanta of oscillation, we define the separable Hilbert space
\begin{equation}\label{eq:H}
\cH=\bigg\{\ket{\psi}=\sum_{n\in\NN} \psi_n \ket{n}~\Big|~\psi_n\in\CC,~\sum_{n\in\NN} |\psi_n|^2 < +\infty   \bigg\}
\end{equation}
equipped with the usual Hermitian product $\bket{\psi|\phi}= \sum_{n\in\NN} {\as \bar\psi_n} \phi_n $ between $\ket{\psi}=\sum_{n\in\NN} \psi_n \ket{n}$ and $\ket{\phi}=\sum_{n\in\NN} \phi_n \ket{n}$, where {\as$\bar z$} denotes the complex conjugate of $z \in \mathbb{C}$.

The set of Hermitian trace-class operators on $\cH$ is denoted by $\Ktr$. Any $\rho\in\Ktr$ is a compact Hermitian operator admitting a spectral decomposition
$$
\rho= \sum_{\mu\geq 1} \lambda_\mu  \ket{\psi_\mu}\bra{\psi_\mu}
$$
with $\{\ket{\psi_\mu}\}_{\mu\geq 1}$ a Hilbert basis of $\cH$, $\lambda_\mu\in\RR$ and $\sum_{\mu\geq 1} |\lambda_\mu| < +\infty$. The state of an open quantum system is described by a \emph{density operator}, which must belong to the set
$$ \cD = \left\{ \rho \in \Ktr~\Big| \sum_{\mu\geq 1} \lambda_\mu = 1, \, ;\; \lambda_\mu \geq 0 \text{ for all } \mu\geq 1 \,\right\} $$
of Hermitian positive semidefinite operators with trace one. An interpretation is that $\lambda_\mu$ gives the probability of the quantum system to be in the corresponding state $\ket{\psi_\mu}$, and more generally $\bra{\psi} \rho \ket{\psi}$ is the probability that the system behaves as if it was in the state $\ket{\psi}$, for any normalized $\ket{\psi}$ on $\cH$.

{\as The classical harmonic oscillator has a state space $\mathbb{R}^2$, and identifying $\mathbb{R}^2$ with $\mathbb{C}$ its solution set is often written $\{ x(t) = \alpha \, e^{i\omega t}  |  \alpha \in \mathbb{C} \}$, with $\alpha$ encoding oscillation amplitude and initial phase, while $\omega$ is the fixed eigenfrequency of the harmonic oscillator. In a reference frame that rotates as $e^{i\omega t}$, the solutions boil down to $\tilde{x}(t) = \alpha$.
For a classical harmonic oscillator with damping and driving input, it is still insightful to write the dynamics in the rotating frame, and now $\tilde{x}(t)$ will vary in time, for instance asymptotically converging towards a unique constant $\tilde{\alpha} \in \mathbb{C}$ if the input is driving the harmonic oscillator at its resonance frequency $\omega$; one can take $\alpha$ real positive by an appropriate choice of phase reference for this input drive.

The model of the quantum harmonic oscillator below is precisely the quantum analog of this situation, i.e.~it describes the evolution of the harmonic oscillator of eigenfrequency $\omega$ and driven at this eigenfrequency\footnote{This focus on drives at the eigenfrequency results from the so-called ``Rotating Wave Approximation''.}, in a reference frame that rotates at this very same frequency.}
The closest quantum equivalent to a classical harmonic oscillator state of amplitude $\alpha \in \mathbb{C}$ is the so-called coherent state
$$\q{\alpha}=e^{-\frac{|\alpha|^2}{2}} \sum_{n=0}^\infty \frac{\alpha^n}{\sqrt{n!}}\ket{n} \, .$$
It is characterized by its invariance under the photon annihilation operator $\ba$, which is defined by $\ba\ket{n}=\sqrt{n}\ket{n-1}$ for $n >0$ and $\ba \ket 0=0$, {\as with domain
${\cal H}_{1} = \{\ket{\psi}=\sum_{n\in\NN} \psi_n \ket{n}~|~\psi_n\in\CC,~\sum_{n\in\NN} n|\psi_n|^2 < +\infty   \}$. One indeed checks that $\q{\alpha} \in {\cal H}_{1}$ and} $\ba \q{\alpha} = \alpha \q{\alpha}$. {\as In this paper, we emphasize that notation $\q{\alpha}$ (eventually with a subscript) will always refer to a coherent state while a Roman letter or numeral like $\ket{n}$ or $\ket{n-1}$ will always refer to a canonical basis state.} For any given integer $k \geq 1$, the aim of the reservoir engineered in \cite{MirrahimiCatComp2014,legthas2015science} is to stabilize a manifold of so-called ``$k$-legged Schr\"odinger cat states''. Such a state is in fact a linear superposition of coherent states, of the form
\begin{equation}\label{eq:defcat}
\q{\psi} = \frac{1}{\vartheta} \sum_{m=1}^k \; e^{i\theta_m} \; \q{\alpha_m} \; , \qquad \alpha_m = \alpha\, e^{2i\pi m/k}
\end{equation}
where the $\theta_m$ are fixed quantum phases; $\alpha$ can be taken real positive by an appropriate choice of reference frame; and $\vartheta$ is a normalization constant that depends on $\alpha$ and $\{\theta_1,...,\theta_k \}$. {\as The coherent superposition of different ``classical'' harmonic oscillator solutions at the same time --- i.e. $\q{\alpha_1} + \q{\alpha_2} \neq \q{\alpha_1+\alpha_2}$ --- is an intrinsic feature of quantum mechanics.}

The reservoir proposed in \cite{MirrahimiCatComp2014,legthas2015science} to stabilize such states of the form \eqref{eq:defcat}, can be modeled by the following Lindblad master equation for the evolution of the quantum harmonic oscillator's state:
\begin{equation}\label{eq:mod}
\tfrac{d}{dt} \rho = \mathfrak{L}_{\as \bL}(\rho) = \bL \rho \bL^\dagger - \frac{1}{2} \bL^\dagger \bL \rho - \frac{1}{2} \rho \bL^\dagger \bL \, \quad \text{with }\;\; \bL=\ba^k-\alpha^k\bI \;\; {\as , \alpha \in \mathbb{R},\; k \in \{1,2,3,...\} \; .}
\end{equation}
{\as For $k=1$, it is a standard result that \eqref{eq:mod} is equivalent to a classical damped and driven harmonic oscillator, whose solution asymptotically converges to the coherent state $\q{\alpha}$. Quantum mechanics allows more complicated damping and drive operators, here thus with $k>1$ and leading to other steady states specific to quantum mechanics.} Note that the operator $\bL$ is unbounded, for any $\alpha \in \mathbb{R}$ and any $k\geq 1${\as . We denote its domain ${\cal H}_{k} = \{
\ket{\psi}=\sum_{n\in\NN} \psi_n \ket{n}~| \psi_n \in \mathbb{C},\; \sum_{n\in\NN} n^{k} |\psi_n|^2 < +\infty \}
\}$, and the domain of ${\bL^\dag \bL}$ is then ${\cal H}_{2k}$. These domains contain the coherent states $\ket{\alpha}$, for any $k \geq 1$.} In this paper we prove the well-posedness and convergence properties of this equation \eqref{eq:mod}, in the infinite-dimensional setting associated to the quantum harmonic oscillator Hilbert space $\cH$ as defined in \eqref{eq:H}.\\

{\as \paragraph{\it Remark about notation:} We have kept the usual physicists' notation for quantum mechanics. For a given Hilbert space $\cal H $, we denote by a ket $\ket{\bullet}$ an element of $\cal H$ representing a quantum state. A bra $\bra{\bullet}$ is an element of the dual space $\cal H^*$. In particular, in finite dimension, one can see a ket $\ket{\bullet}$ as a column vector while a bra $\bra{\bullet}$ represents a row vector. For any operator $O$, the scalar product $\bra{x}O^\dag O \ket{x}$ is equivalent to the quadratic norm $||O \ket{x}||^2$ in the Hilbert space sense.
For the quantum harmonic oscillator, in a less abstract formulation one may choose $\mathcal{H}=L^2(\mathbb{R})$, with $\mathbb{R}$ corresponding e.g.~to position of a mechanical oscillator. The canonical Hilbert basis is then formed by the Hermite functions, describing the distribution in position for an eigenfunction of the Schr\"odinger equation for the quantum harmonic oscillator. In this representation, the annihilation operator $\ba=\frac{1}{\sqrt{2}}\left( \partial_x + x \right)$ and equation \eqref{eq:mod} can be interpreted as a partial differential equation with order $2k$ derivation in $x$ -- the standard Laplace operator is retrieved for $k=1$.}


\section{Well-posedness}\label{sec:WP}

In the sequel, for any $\rho= \sum_{\mu\geq 1} \lambda_\mu  \ket{\psi_\mu}\bra{\psi_\mu} \in \Ktr$ we use the notation
$$
\rho^+ =  \sum_{\mu\geq 1} \max(0,\lambda_\mu)  \ket{\psi_\mu}\bra{\psi_\mu} , \quad \rho^-=\sum_{\mu\geq 1} \max(0,-\lambda_\mu)  \ket{\psi_\mu}\bra{\psi_\mu} \, .
$$
Thus $\rho=\rho^+ - \rho^-$ and $|\rho|= \rho^+ + \rho^-$. Equipped with the trace-norm
$$
\|\rho\|_{tr} = \tr{|\rho|}=\sum_{\mu=1}^{\infty} |\lambda_\mu| \, ,
$$
$\Ktr$ is a Banach space.
For any $\rho\in\Ktr$ and any bounded operator $\bB$ on $\cH$,  the operators $\bB \rho$ and $\rho \bB$ are trace-class operators and
\begin{equation}\label{eq:rhoB}
\tr{ \bB \rho } = \tr{\rho \bB }, \quad \tr{\bB \rho} \leq \tr{|\bB \rho|} = \| \bB \rho \|_{tr} \leq
\|\bB\|~ \tr{|\rho|} = \| \bB \|\; \|\rho\|_{tr}
\end{equation}
with the standard induced operator norm
$$\| \bB \| = \max\{ \| \bB \ket{\psi} \|_{\cH} \; \Big| \;\ket{\psi} \in \cH, \; \| \ket{\psi} \|_{\cH}=1 \}\,.$$
We denote by $\cH^f$ the sub-space of $\cH$ associated to $\ket\psi$ involving a finite number of photons {\as i.e. a finite linear combination of vectors in the canonical basis}:
$$
\cH^f=\bigg\{\ket{\psi}=\sum_{n\in\NN} \psi_n \ket{n}~\Big|~\psi_n\in\CC,~\exists \bar n \text{ such that } \psi_n=0 \text{ for } n>\bar n  \bigg\}\, .
$$
$\cH^f$ is dense in $\cH$.
We denote also by $\Kf$, the subspace of $\Ktr$ of operators whose range is included in a vector space spanned by a finite number {\as of vectors in the canonical basis}:
$$
\Kf=\left\{ \sum_{n=1}^{\bar n}\;\sum_{n'=1}^{\bar n}\; f_{n,n'} \ket{n}\bra{n'}~ \Big| ~f_{n,n'}\in\mathbb{C},\;\bar n \in \mathbb{N}_{>0} \right\}\, .
$$
$\Kf$ is dense in the Banach space $\Ktr$. More details about these operator spaces can be found in~\cite{TarasovBook2008}. {\as The domain of an operator ${\bf A}$ on $\cH$ is the set ${\cal D_{\bf A}} =\{ \q{\psi} \in \cH ~\Big|~ {\bf A}\q{\psi} \in \cH \} \subset \cH$; we will not specify it explicitly for each case. When manipulating operators in the sequel, we often investigate the behavior of the operators on $\cH^f$ and conclude by density. In such discussion, the reader can thus consider the restriction of the operator domains to $\cH^f$ in order to ensure that everything is well-defined. When we are not restricted to $\cH^f$, commutation relations must also be treated with care, e.g.~ensuring that the assumptions around \eqref{eq:rhoB} are satisfied or using some other arguments which we then specify.}

The photon-number operator is defined by $\bN = \sum_{n\in\NN} \, n \, \ket{n}\bra{n}$ and it satisfies $\ba^\dagger \ba = \bN$, whereas $\ba \ba^\dagger = \bN + \bI = \sum_{n\in\NN} \, (n+1)\, \ket{n}\bra{n}$. For any $\nu\in\NN$, we denote $(\bN-\nu \bI)^+$ the Hermitian operator defined by $(\bN-\nu \bI)^+\ket n= (n-\nu)\ket n$, for $n> \nu$ and $(\bN-\nu \bI)^+\ket n=0$ for $n\in\{0,\ldots,\nu\}$. Then explicit computations show that, {\as in ${\cal H}_{2k}$, we have}
\begin{equation}\label{eq:M}
[\bL,\bL^\dag]\equiv \bL\bL^\dag - \bL^\dag \bL =  \ba^k(\ba^\dag)^k - (\ba^\dag)^k \ba^k =  \bM
\end{equation}
where $\bM= (\bN+ \bI)(\bN+2\bI) \ldots (\bN+k\bI)\;\, - \;\,\bN (\bN- \bI)^+ \ldots (\bN-(k-1)\bI)^{+}$ is a positive Hermitian operator {\as of domain ${\cal H}_{2k}$}, unbounded but diagonal in the Fock basis
$\{\ket{n}\}_{n\in\NN}$.

\begin{lem} \label{lem:MaxMonotone}
The operator $\bL^\dag \bL$   {\as with domain ${\cal H}_{2k}$} admits a spectral decomposition $\bL^\dag\bL = \sum_{\mu=1}^{\infty}d_\mu  \ket{g_\mu}\bra{g_\mu}$  where $\big(\ket{g_\mu}\big)_{\mu \geq 1}$ is an Hilbert basis of $\cH$ and $d_\mu \geq 0$.
\end{lem}
\begin{proof}
We construct the spectral decomposition of $\bL^\dag \bL$ from {\as the one of the resolvent, i.e.~the} inverse of $\bI +\lambda \bL^\dag\bL$ for $\lambda >0$ small enough. {\as Let us define $\bR=\bI + \lambda\bigg(\bN (\bN- \bI)^+ \ldots (\bN-{\as (k-1)\bI})^{+} + \alpha^{2k} {\as \bI} \bigg)$ with domain ${\cal H}_{2k}$. Obviously, $\bR$ is Hermitian and positive, as by definition it is diagonal in the Fock basis $\{\ket{n}\}_{n\in\NN}$. This diagonal form allows to define the inverse $\bR^{-1}: \cH \rightarrow \cH$ as an operator on $\cH$, diagonal in the Fock basis, whose spectrum is bounded and decays to zero; thus $\bR^{-1}$ is compact. This inverse satisfies $\bR \bR^{-1} = \bI$ on $\cH$ and $\bR^{-1} \bR = \bI$ on ${\cal H}_{2k}$.

The operator $\bI +\lambda \bL^\dag\bL$ has the same domain ${\cal H}_{2k}$ for any $\lambda >0$. Then noting that
   $
   \bL^\dag \bL = \bN (\bN- \bI)^+ \ldots (\bN-{\as (k-1)\bI})^{+} + \alpha^{2k}{\as \bI} - \alpha^{k} ((\ba^\dag)^k + \ba^k) \, ,
   $
we can write
\begin{equation}\label{eq:I+lL}
\bI +\lambda \bL^\dag\bL = \bR - \lambda \alpha^k  \Big((\ba^\dag)^k + \ba^k\Big) = \bigg(\bI - \lambda \alpha^k  \Big((\ba^\dag)^k + \ba^k\Big)  \bR^{-1}\bigg) \bR
\end{equation}
where the domain of $\bI - \lambda \alpha^k  \Big((\ba^\dag)^k + \ba^k\Big)$ is also ${\cal H}_{2k}$.}
We now work towards inverting the product on the right side of~\eqref{eq:I+lL}.\vspace{2mm}

Take $\ket\psi=\sum_{n\in\NN} \psi_n \ket n $ in $\cH^f {\as \subset \cH_{2k}}$.
Then
\begin{multline*}
 ((\ba^\dag)^k + \ba^k)\ket\psi
 = \sum_{n\geq 0}  \sqrt{(n+1)\ldots (n+k)}\psi_{n} \ket{n+k}
 +\sum_{n\geq k} \sqrt{n(n-1)\ldots (n-k+1)}\psi_{n}\ket{n-k}
\end{multline*}
and a few computations lead to the bound
\begin{equation}\label{eq:bound1}
 \bket{\psi|((\ba^\dag)^k + \ba^k)^2|\psi } \leq 2\sum_{n\geq 0}(  (n+k)^k + n^k ) |\psi_n|^2
 .
\end{equation}
Clearly $\ket{\phi}=\bR^{-1}\ket \psi= \sum_{n\geq 0} \phi_n \ket n$ belongs also to $\cH^f$ since
$
\phi_n= \frac{\psi_n}{1+ {\as \lambda (n(n-1)\ldots(n-k+1)} + \alpha^{2k})} \, .
$
Inserting this into \eqref{eq:bound1} gives for any $\ket\psi\in\cH^f$,
\begin{multline*}
  \lambda^2 \alpha^{2k}   \bket{\psi\left|\bR^{-1}\Big((\ba^\dag)^k + \ba^k\Big)^2\bR^{-1}\right|\psi}
\leq
\sum_{n\geq 0} \frac{ 2 \lambda^2\alpha^{2k} ( (n+k)^k + n^k) }{\left(1+ \lambda {\as (n(n\text{-}1)\ldots(n\text{-}k+1)} + \alpha^{2k})\right)^2} |\psi_n|^2
.
\end{multline*}
A rough estimation shows that there exists a constant $c>0$ such that $\forall n\in\NN$ {\as and for all $\lambda > 0$}, we have
$$
\frac{ 2 \lambda^2\alpha^{2k} ( (n+k)^k + n^k) }{\left(1+ \lambda ({\as n(n-1)\ldots(n-k+1)} + \alpha^{2k})\right)^2} \leq  c \lambda
.
$$
{\as Thus for all $\ket\psi \in\cH^f$ we have
$$
\ket\psi \in\cH, \quad \lambda^2 \alpha^{2k}   \bket{\psi\left|\bR^{-1}\Big((\ba^\dag)^k + \ba^k\Big)^2\bR^{-1}\right|\psi}
\leq c\lambda \bket{\psi|\psi}.
$$
From this, we can define $\lambda \alpha^k   \Big((\ba^\dag)^k + \ba^k\Big)\bR^{-1}$ as a bounded operator on $\cH^f$, and for $\lambda < 1/c$ its norm is strictly less than one. By density of $\cH^f$ in $\cH$, defining a Cauchy sequence in $\cH^f$ for $j=1,2,... : \ket{\psi_j} \rightarrow \ket{\psi} \in \cH$ with $\bket{ n | \psi_j} = \bket{ n | \psi}$ for all $n \leq j$, $\bket{ n | \psi_j} = 0$ for all $n > j$, we can extend the domain of $\lambda \alpha^k   \Big((\ba^\dag)^k + \ba^k\Big)\bR^{-1}$ to all of $\cH$, with still a norm $<1$ for $\lambda < 1/c$.}

This implies that  $\bigg(\bI - \lambda \alpha^k  \Big((\ba^\dag)^k + \ba^k\Big)\bR^{-1}\bigg)$ is bounded and admits a bounded inverse on $\cH$. We can then define
$$
(\bI +\lambda \bL^\dag\bL)^{-1} =\bR^{-1} \bigg(\bI - \lambda \alpha^k   \Big((\ba^\dag)^k + \ba^k\Big)\bR^{-1}\bigg)^{-1}
$$
which is the product of a bounded operator with a compact operator, hence it is compact {\as with domain $\cH$}. Thus  $(\bI +\lambda \bL^\dag\bL)^{-1}$  is a compact  Hermitian operator:  it admits a spectral decomposition. By a $\cH^f$ density argument, we have $0 \leq \bra{\psi} (\bI +\lambda \bL^\dag\bL)^{-1} \ket{\psi} \leq \bket{\psi | \psi}$ for all $\q{\psi} \in \cH$. Taking all this together, we can then write
$$
(\bI +\lambda \bL^\dag\bL)^{-1} = \sum_{\mu \geq 1} s_\mu \ket{g_\mu}\bra{g_\mu}
$$
with $0 \leq s_\mu \leq 1$ and $\{\ket{g_\mu}\}_{\mu \geq 1}$ a Hilbert basis {\as of $\cH$}. We conclude with $d_\mu =(1/s_\mu  -1)/\lambda$.

\end{proof}

Via the spectral decomposition of Lemma~\ref{lem:MaxMonotone}, $\bS = \sqrt{\bI+ \bL^\dag\bL}$ is well defined by
$$
\bS =  \sqrt{\bI+ \bL^\dag\bL}= \sum_{\mu=1}^{\infty} \sqrt{1+d_\mu} \,  \ket{g_\mu}\bra{g_\mu}\, .
$$
Denote by
$$
\KL= \left\{ \rho \in\Ktr~\Big|~ \tr{\left|\bS  \rho \bS \right| }  < +\infty \right\}
.
$$
We have the following lemma.
\begin{lem}\label{lem:NormL}
$\KL$ equipped with the norm
\begin{equation}\label{eq:NormL}
  \Lnorm{\rho} =\tr{\left|\bS  \rho \bS \right| }
\end{equation}
is a Banach space. Moreover $\rho \in \KL$  implies $\bL\rho\bL^\dag \in\Ktr$.
\end{lem}
\begin{proof}
The first statement holds since the operation {\as $\rho \mapsto \bS^{-1} \rho \bS^{-1}$} is an isometry mapping the Banach space $\Ktr$, equipped with the trace norm, to the space $\KL$ equipped with the norm $\Lnorm{\rho}$.

{\as For the second statement, consider the operator $\bBl = \bL (\bI+\bL^\dag\bL)^{-1/2}$.  Consider the spectral decomposition of $\bL^\dag\bL$ given in lemma~\ref{lem:MaxMonotone}: for any $\mu$, we have $(\bI+\bL^\dag\bL)^{-1/2}\ket{g_\mu} = \sqrt{\frac{1}{1+d_\mu}} \ket{g_\mu}$. Thus
$$
\|\bBl\ket{g_\mu}\|^2=
\bra{g_\mu}(\bI+\bL^\dag\bL)^{-1/2} \bL^\dag \bL (\bI+\bL^\dag\bL)^{-1/2}\ket{g_\mu}=\frac{d_\mu}{1+d_\mu}
$$
Since $(\ket{g_\mu})_{\mu\geq 1}$ is a Hilbert basis of $\cH$, $\bBl$ is a bounded operator with $\|\bBl\|\leq 1$.
Similarly $\bBl^\dag$ is also a bounded operator of norm  not exceeding one: this results from lemma~\ref{lem:MaxMonotone} and from $\bBl \bBl^\dag= (\bI+\bL\bL^\dag)^{-1}\bL\bL^\dag$ based on  the identity
$(\bI+\bL^\dag\bL)^{-1} \bL^\dag = \bL^\dag (\bI+\bL\bL^\dag)^{-1}$

Take now   $\rho \in \KL$: the operator  $\sigma = (\bI+\bL^\dag\bL)^{1/2}\rho (\bI+\bL^\dag\bL)^{1/2}$ is  thus a trace class operator. Since
$\bL\rho\bL^\dag = \bBl \sigma \bBl^\dag$
with $\bBl$ and $\bBl^\dag$ bounded operators with norm less or equal to one, we have following~\eqref{eq:rhoB},  $\|\bL\rho\bL^\dag\|_{tr} \leq  \|\sigma\|_{tr}$. }
\end{proof}

The above considerations put us on track towards the following result.

\begin{thm}\label{thm:Cauchy}
Consider the Cauchy problem~\eqref{eq:mod} associated to  the super-operator $\fA$,
$$
\rho \mapsto \fA(\rho)=  (\bL^\dag\bL \rho + \rho \bL^\dag\bL )/2 - \bL\rho\bL^\dag
.
$$
For any integer $k >0$, any real $\alpha >0$ {\as and any  $\rho_0$ in the domain of $\fA$, there exists a unique $C^1$ function $[0,+\infty[\ni t \mapsto \rho(t)\in\KL$, such that $\rho(t)$ belongs to the domain of $\fA$ for all $t\geq 0$
and solves the initial value problem~\eqref{eq:mod} with $\rho(0)=\rho_0$.}
Moreover, for all $t$ positive,  $\tr{\rho(t)} = \tr{\rho_0}$,  $\Lnorm{\rho(t)} \leq  \Lnorm{\rho_0}$ {\as and $\Lnorm{\fA(\rho(t))}\leq \Lnorm{\fA(\rho_0)}$. If $\rho_0$ is non-negative then $\rho(t)$ remains also non negative. }
\end{thm}

The proof is based on the Hille-Yosida theorem, recalled in the appendix, {\as and which indeed ensures continuous derivability in $\KL$. To apply Hille-Yosida we must prove} that the unbounded super-operator $\fA$ on $\KL$ is $m$-accretive. Its domain $D(\fA)$ is dense since it contains $\Kf$. The essential difficulty is to prove that for any $\lambda >0$, $\bI+\lambda \fA$ is a bijection from $D(\fA)$ into $\KL$, with $(\bI + \lambda \fA)^{-1}$ a bounded linear operator on $\KL$ with norm less or equal to $1$.
The proof is partially  inspired by~\cite{Davie1977RoMP}: it is decomposed into the successive lemmas~\ref{lem:Ar} to~\ref{lem:accretive}. One of the key and original arguments, {\as used e.g.~in \eqref{eq:trLPIL}}, is that the commutator $[\bL,\bL^\dag]=\bM$ defines a non-negative Hermitian (unbounded) operator, see \eqref{eq:M}.

\begin{lem} \label{lem:Ar} For any $f\in\KL$ (resp. $ f\in \Ktr$), any $\lambda >0$ and $r\in[0,1)$, there exists a unique $\rho_{r} \in\KL$ (resp. $\rho_{r} \in\Ktr$)  solution of
\begin{equation}\label{eq:Ar}
\tfrac{\bI+\lambda \bL^\dag\bL}{2} \rho_r + \rho_r \tfrac{\bI+\lambda \bL^\dag\bL}{2} = f + r \lambda \bL \rho_r \bL^\dag
\, .
\end{equation}
Moreover we have $\Lnorm{\rho_{r}}\leq \Lnorm{f}$ (resp. $\|\rho_r\|_{tr} \leq \|\rho_r\|_{tr}$).
When  additionally $f\geq 0$, we have  $0 \leq \rho_{r_1} \leq \rho_{r_2}$ for $0\leq r_1\leq r_2 <1$ ;
\end{lem}
We will not detail below the proof  in the $\Ktr$ space. It relies on similar  but  simpler  arguments to the ones used  in the proof below for the $\KL$ space.
\begin{proof}
{\as The diagonalization, boundedness and positivity of $(I+\lambda \bL^\dag \bL)^{-1}$ as established in the proof of Lemma~\ref{lem:MaxMonotone}, allows it to generate} a strongly continuous semigroup $\{e^{- s ( \bI+\lambda \bL^\dag\bL)}\}_{s\geq 0}$ of contractions on $\cH$ {\as with $\left\| e^{- s ( \bI+\lambda \bL^\dag\bL)}\right\|\leq e^{-s}$ for all $s\geq 0$}. Hence for any $\xi\in\KL$, equivalently for any $\phi = \bS \xi \bS \in \Ktr$, there exists a unique solution $\rho\in\KL$ of the Sylvester equation
   $$
   \left(\tfrac{\bI+\lambda \bL^\dag\bL}{2}\right)\rho + \rho \left(\tfrac{\bI+\lambda \bL^\dag\bL}{2}\right)= \xi
   = \bS^{-1} \phi \bS^{-1}$$
and given by the usual formula{\as~\cite[section VII-2]{bhatiaBook97} }
\begin{eqnarray}\label{eq:SylvSol}
\rho & = & \Pi(\xi) \;=\int_{0}^{+\infty} e^{- s (\bI+\lambda \bL^\dag\bL)/2} \, \xi \, e^{- s (\bI+\lambda \bL^\dag\bL)/2}\, ds \\
\nonumber & = & \int_{0}^{+\infty} {\as e^{- s (\bI+\lambda \bL^\dag\bL)/2} \; \bS^{-1} \phi \bS^{-1} \; e^{- s (\bI+\lambda \bL^\dag\bL)/2}} \, ds
\end{eqnarray}
{\as where the  integral is  defined as  a  simple  Riemann integral of functions valued in the  Banach space of bounded operators on $\cH$.}
Note that the semigroup generator $e^{- s (\bI+\lambda \bL^\dag\bL)/2}$ commutes with $\bS^{-1}$ {\as and both operators are bounded}, so in fact $\Pi(\phi) = \Pi(\bS \xi \bS) = {\as \bS \Pi(\xi) \bS}$.
{\as Replacing the left hand side by $\xi$ and $\rho_r$ by $\Pi(\xi)$ in \eqref{eq:Ar} gives}
$$\xi=  f +  r \lambda \bL \Pi(\xi) \bL^\dag = f + r \fB(\xi) \quad \text{with } \fB(\xi)=  \lambda \bL \Pi(\xi) \bL^\dag \, .$$
We will conclude the existence and uniqueness proof by showing that $\xi \mapsto f + r\fB(\xi)$ is a strict contraction for the $\Lnorm{\cdot}$-norm on $\KL$. Each part of the proof first considers the case of positive operators, then (between $\rhd$ $\lhd$) adapts it to arbitrary ones. \vspace{2mm}

\noindent $\bullet$ \emph{Step 1:} Contraction, monotonicity of $\Pi(\xi)$ and some trace estimates.

From \eqref{eq:SylvSol} it is obvious that $\Pi(\xi) \geq 0$ for any $\xi \geq 0$ in $\KL$, hence by linearity $\Pi(\xi_1)\leq \Pi(\xi_2)$ as soon as $\xi_1 \leq \xi_2$ in $\KL$, i.e.~$\Pi$ is monotone. {\as Consider
\begin{multline*}
\Lnorm{\bL\Pi(\xi)\bL^\dag} =
\tr{\bS \bL\Pi(\xi)\bL^\dag \bS} \\ = \int_{0}^{+\infty} \tr{\bS \bL e^{- s (\bI+\lambda \bL^\dag \bL )/2}  \bS^{-1} \phi \bS^{-1}
   e^{- s (\bI+\lambda \bL^\dag \bL)/2} \bL^\dag  \bS} \, ds
\end{multline*}
Since $\bS \bL e^{- s (\bI+\lambda \bL^\dag\bL)/2} \bS^{-1}$ and its Hermitian conjugate  are  bounded operators we have thanks to \eqref{eq:rhoB}:
\begin{multline*}
  \tr{\bS \bL e^{- s (\bI+\lambda \bL^\dag \bL )/2}  \bS^{-1} \phi \bS^{-1}
   e^{- s (\bI+\lambda \bL^\dag \bL)/2} \bL^\dag  \bS}
   \\= \tr{\bS^{-1}
   e^{- s (\bI+\lambda \bL^\dag \bL)/2} \bL^\dag  \bS^2\bL e^{- s (\bI+\lambda \bL^\dag \bL )/2}  \bS^{-1} \phi }
   .
\end{multline*}
Consider now  the non negative Hermitian operator $\bL^\dag  \bS^2\bL= \bL^\dag (\bI+\bL^\dag\bL)\bL$. Since $\bL^\dag \bL \leq \bL\bL^\dag$ we have
$$
\bL^\dag  \bS^2\bL \leq  \bL^\dag (\bI+\bL\bL^\dag)\bL= (\bI+ \bL^\dag \bL) \bL^\dag\bL=\bL^\dag\bL(\bI+ \bL^\dag \bL)
.
$$
Thus, for each $s\geq 0$,  we have
\begin{multline*}
  e^{- s (\bI+\lambda \bL^\dag \bL)/2} \bL^\dag  \bS^2\bL e^{- s (\bI+\lambda \bL^\dag \bL )/2}  \bS^{-1}
  \\ \leq
  \bS^{-1}  e^{- s (\bI+\lambda \bL^\dag \bL)/2} (\bI + \bL^\dag \bL)\bL^\dag\bL   e^{- s (\bI+\lambda \bL^\dag \bL )/2}  \bS^{-1}
   \\
   = \bL^\dag \bL e^{- s (\bI+\lambda \bL^\dag \bL)}=e^{- s (\bI+\lambda \bL^\dag \bL)}\bL^\dag \bL
   .
\end{multline*}
Since $\phi$ is trace class and nonnegative, its spectral decomposition  reads $\phi= \sum_{\mu} \sigma_\mu \ket{\phi_\mu}\bra{\phi_\mu}$ where $(\ket{\phi_\mu})_{\mu\geq 1}$ forms a Hilbert basis, $\sigma_\mu\geq 0$ for each $\mu\geq 1$ and $\sum_\mu \sigma_\mu = \tr{\phi} < +\infty$. Thus, since all terms are nonnegative in the  series below, we have
\begin{multline*}
\Lnorm{\bL\Pi(\xi)\bL^\dag}
= \int_{0}^{+\infty}\left(\sum_{\mu\geq 1}  \sigma_\mu  \bra{\phi_\mu}\bS^{-1}
   e^{- s (\bI+\lambda \bL^\dag \bL)/2} \bL^\dag  \bS^2\bL e^{- s (\bI+\lambda \bL^\dag \bL )/2}  \bS^{-1}  \ket{\phi_\mu}\right) ds
   \\
\leq \sum_{\mu\geq 1}  \sigma_\mu \left\langle \phi_\mu \left| \int_0^{+\infty}\bL^\dag \bL e^{- s (\bI+\lambda \bL^\dag \bL)}~ds \right| \phi_\mu \right\rangle
.
\end{multline*}
Since $\int_0^{+\infty}\bL^\dag \bL e^{- s (\bI+\lambda \bL^\dag \bL)}~ds= \bL^\dag\bL (\bI + \lambda \bL^\dag \bL)^{-1}$ (use the spectral decomposition of $\bL^\dag\bL$), we have
\begin{multline}\label{eq:trLPIL}
 \Lnorm{\bL\Pi(\xi)\bL^\dag}
\leq \sum_{\mu\geq 1}  \sigma_\mu \left\langle \phi_\mu \left| \bL^\dag\bL (\bI + \lambda \bL^\dag \bL)^{-1} \right| \phi_\mu \right\rangle \\
= \tr{\bL^\dag\bL (\bI + \lambda \bL^\dag \bL)^{-1} \phi} \leq \tfrac{1}{\lambda} \tr{\phi}
= \tfrac{1}{\lambda} \Lnorm{\xi}
.
\end{multline}
}
We can repeat a similar argument to get
\begin{eqnarray}
\Lnorm{\Pi(\xi)} & = & \tr{\Pi(\phi)} =  {\as \tr{(\bI+\lambda \bL^\dag\bL)^{-1} \phi }}
\label{eq:trPI} \leq \tr{\phi} = \Lnorm{\xi} \, .
\end{eqnarray}

\begin{minipage}{140mm}
$\rhd$
For $\xi \not\geq 0$ in $\KL$, writing $\phi=\phi^+ - \phi^-$ we have $\phi^+,\phi^-\in\Ktr$ nonnegative and such that $|\phi| = |\bS \xi \bS| =\phi^+ + \phi^-$.
Then we get, using the above, linearity and monotonicity of $\Pi$, and the triangular inequality for the trace-norm:
\begin{eqnarray*}
  \Lnorm{\Pi(\xi)} & = & \tr{ |\Pi(\phi^{+}) - \Pi(\phi^-)| } \leq \tr{\Pi(\phi^{+}) } + \tr{\Pi(\phi^-)}
  \\ & \leq & \tr{\phi^+} + \tr{\phi^-} = \Lnorm{\xi}. \hspace{65mm} \lhd
\end{eqnarray*}
\end{minipage}\vspace{2mm}
\newline Thus $\Pi$ is a (non-strict) contraction for the $\Lnorm{\cdot}$-norm on $\KL$.\vspace{2mm}

\noindent $\bullet$ \emph{Step 2:} Contraction of $\fB(\xi)$ and consequences.

Let us prove that $\xi \mapsto \fB(\xi)=  \lambda \bL \Pi(\xi) \bL^\dag$ is a (non-strict) contraction for the $\Lnorm{\cdot}$-norm on $\KL$. For any $\xi\in\KL$ nonnegative, $\phi$ and $\fB(\xi)$ are nonnegative. We deduce from~\eqref{eq:trLPIL},\eqref{eq:trPI} that $\fB(\xi)$ belongs to $\KL$ with
\begin{equation}\label{eq:TrB}
\Lnorm{\fB(\xi)} \leq \lambda {\as \tr{(\bL^\dag\bL)(\bI+\lambda \bL^\dag\bL)^{-1} \phi }}
= \tr{\phi} - \tr{\Pi(\phi)} \; \leq \; \tr{\phi}=\Lnorm{\xi} \; .
\end{equation}

\begin{minipage}{140mm}
$\rhd$ For $\xi \in \KL$ arbitrary, decompose $\phi\not\geq 0$ into $\phi=\phi^+ - \phi^-  \in \Ktr$ and similarly to above we get:
\begin{eqnarray*}
\Lnorm{\fB(\xi)} &=& \tr{\left| \bS \fB(\bS^{-1} \phi^+ \bS^{-1}) \bS - \bS \fB(\bS^{-1} \phi^- \bS^{-1}) \bS \right|} \\
 & \leq & \tr{\left| \bS \fB(\bS^{-1} \phi^+ \bS^{-1}) \bS \right|} + \tr{\left| \bS \fB(\bS^{-1} \phi^- \bS^{-1}) \bS \right|} \\
 & \leq & \tr{\phi^+} + \tr{\phi^-} \; = \; \tr{|\phi|} \; = \; \Lnorm{\xi} \; ,
\end{eqnarray*}
since on the second line we can drop the absolute value.\hfill $\lhd$
\end{minipage}
\vspace{2mm}

This proves (non-strict) contraction of $\fB(\xi)$. Thus $\xi \mapsto f + r\fB(\xi)$ is a strict contraction on $\KL$ as soon as $r \in [0,1)$. Consequently, it admits a unique fixed point $\xi_r\in\KL$ given by the absolutely converging series
\begin{equation}\label{eq:Cr}
\xi_r=\fC_r(f)= \sum_{s=0}^{+\infty} r^{s}\fB^{s} (f).
\end{equation}
This justifies a posteriori that taking $\xi \in \KL$ for the Sylvester equation yields a valid result.
The solution $\rho_r$ is then given by $\rho_r= \Pi(\fC_r(f))$.\vspace{2mm}

\noindent $\bullet$ \emph{Step 3:} There remains to prove the  inequality.

First take $f \geq 0$ and  $0 \leq r_1\leq r_2 <1$. Then by \eqref{eq:Cr} we have $0 \leq f \leq \fC_{r_1}(f) \leq \fC_{r_2}(f)$, in particular $\xi_r \geq 0$ for all $r \in [0,1)$. Since $\Pi$ is monotone, we have
$0 \leq \Pi(f) \leq \rho_{r_1} \leq \rho_{r_2}$ which proves the {\as second} inequality. Moreover, we have
\begin{eqnarray*}
\tr{\bS \rho_r \bS} &=& \tr{\bS \Pi(\xi_r) \bS} = \tr{\Pi(\phi_r)} \leq \tr{\phi_r} - \tr{\bS \fB(\xi_r) \bS} \\
& = & \tr{\bS f \bS} + (r-1) \tr{\bS \fB(\xi_r) \bS} \leq \tr{\bS f \bS} \; .
\end{eqnarray*}
At the end of the first line we have used \eqref{eq:TrB}; the next equality comes from the definition $\phi_r = \bS \xi_r \bS = \bS(f + r \fB(\xi_r))\bS$ where all terms have been proved to be trace-class, and the final inequality holds because $f\geq 0$ implies $\xi_r \geq 0$ and thus $\fB(\xi_r) \geq 0$. This would conclude the proof if we impose $f \geq 0$.\vspace{2mm}

\begin{minipage}[b]{140mm}
$\rhd$
To prove that $\tr{|\bS \rho_r \bS|} \leq \tr{|\bS f \bS|}$ for $f \not\geq 0$ in $\KL$, define $g= \bS f \bS \in \Ktr$ and decompose $g = g^+ - g^-$ with $g^+, g^- \geq 0$. Then define $\pos{f} = \bS^{-1} g^+ \bS^{-1}$ and the associated solution
$\pos{\rho_r}$ of \eqref{eq:Ar}, and similarly for $\neg{f}$ and $\neg{\rho_{r}}$. By construction, $\pos{f},\neg{f},\pos{\rho_r},\neg{\rho_r}$ are nonnegative and belong to $\KL$. Note however that nothing guarantees that e.g.~$\pos{\rho_{r}}=\rho_r^+$, such that although $\rho_r = \pos{\rho_{r}} - \neg{\rho_{r}}$ by linearity, possibly $|\rho_r| = \rho_r^+ + \rho_r^- \neq \pos{\rho_{r}} + \neg{\rho_{r}}$. The triangular inequality for the trace-norm nevertheless guarantees
\begin{eqnarray*}
\Lnorm{\rho_r} &\!\!=\!\!&\Lnorm{\pos{\rho_{r}} - \neg{\rho_{r}}} \leq \Lnorm{\pos{\rho_{r}}} + \Lnorm{\neg{\rho_{r}}} = \tr{\bS \pos{\rho_{r}} \bS} + \tr{\bS \neg{\rho_{r}} \bS}\\
& \leq & \tr{\bS \pos{f} \bS} + \tr{\bS \neg{f} \bS} = \tr{g^+ + g^-} = \tr{|g|} = \Lnorm{f} \, .
\end{eqnarray*}
The second inequality is obtained thanks to the property $\tr{\bS \rho_r \bS} \leq \tr{\bS f \bS}$ just proved for $f\geq 0$. The equalities use linearity and positivity. \hfill $\lhd$
\end{minipage}
\end{proof}


\begin{lem} \label{lem:ConvRhor}
Take a nonnegative $f \in \KL$  and consider $\rho_r\in\KL$ given by Lemma~\ref{lem:Ar}. Then, for $r$ tending  to $1^-$,  $\rho_r$ converges {\as towards some $\rho$ in $\KL$. Moreover this $\rho$ is a solution of
\begin{equation}\label{eq:A}
\rho + \lambda \fA(\rho) = \tfrac{\bI+\lambda \bL^\dag\bL}{2} \rho + \rho \tfrac{\bI+\lambda \bL^\dag\bL}{2} - \lambda \bL \rho \bL^\dag = f
,
\end{equation}
$ \rho \geq 0$ and}
$
\Lnorm{\rho}=\tr{\bS  \rho \bS  } \leq  \tr{\bS  f \bS }=\Lnorm{f} \, .
$
\end{lem}

\begin{proof}
Take an increasing sequence $\{r_k\}_{k\in\NN}$ in $[0,1)$ converging to $1$. From
$\Lnorm{\rho_{r_{k_1}}} \leq \Lnorm{\rho_{r_{k_2}}} \leq \Lnorm{f}$ for $k_1 \leq k_2$, as proved in Lemma~\ref{lem:Ar}, we deduce that with $s_k = \Lnorm{\rho_{r_k}} = \tr{\bS \rho_{r_k} \bS}$, the sequence $\{s_k \}_{k \in \NN}$ is positive, increasing and bounded by $\Lnorm{f}$. Thus it converges.
Moreover, Lemma~\ref{lem:Ar} proves that $\rho_{r_{k_2}}-\rho_{r_{k_1}} \geq 0$ for $k_1 \leq k_2$, such that
$$\Lnorm{\rho_{r_{k_2}}-\rho_{r_{k_1}}} = \tr{\bS (\rho_{r_{k_2}}-\rho_{r_{k_1}}) \bS} = s_{k_2} - s_{k_1}$$
and thus $\{\rho_{r_k}\}_{k\in\NN}$ is a Cauchy sequence in the Banach space $\KL$ equipped with the $\Lnorm{\cdot}$-norm. Thus, it converges to some limit $\rho\in\KL$. This $\rho$ is independent of the increasing sequence $r_k$ tending to $1$, and $\Lnorm{\rho}$ is the limit of the sequence $\{s_k \}_{k \in \NN}$, which is bounded by $\Lnorm{f}$.

For any $r\in[0,1)$, we have defined $\rho_r$ to be solution of
$$
 \tfrac{\bI+\lambda \bL^\dag\bL}{2} \rho_r + \rho_r \tfrac{\bI+\lambda \bL^\dag\bL}{2} - \lambda \bL \rho_r \bL^\dag  = f - (1-r) \lambda \bL \rho_r \bL^\dag
.
$$
But $\bL \rho_r \bL^\dag$ is trace-class (see Lemma \ref{lem:NormL}) and it converges to $ \bL \rho \bL^\dag$ in the trace-norm topology. Hence the left-hand side is also trace-class, and it converges to $\tfrac{\bI+\lambda \bL^\dag\bL}{2} \rho + \rho \tfrac{\bI+\lambda \bL^\dag\bL}{2} - \lambda \bL \rho \bL^\dag $ in the trace-norm topology. Thus by taking the limit in $\Ktr$ of the above equality, we get~\eqref{eq:A}.
\end{proof}

The final Lemma proves the bijection property of $\bI + \lambda \fA$, with $\| (\bI + \lambda \fA)^{-1} \| \leq 1$,  such that we can conclude with the Hille-Yosida (Theorem \ref{thm:HillYosida} recalled in appendix).
\begin{lem} \label{lem:accretive}
Take $f\in\KL$ and $\lambda\geq 0$. Then there exists a unique $\rho$ in the domain of $\fA$ and  solution of
$\rho + \lambda \fA(\rho) =f$, with $\Lnorm{\rho} \leq \Lnorm{f}$. Moreover if $f$ is non negative, then $\rho$ is also non negative.
\end{lem}
\begin{proof}
The existence of $\rho$ and bound $\Lnorm{\rho} \leq \Lnorm{f}$ are just a variation of Lemma~\ref{lem:ConvRhor} where we drop the assumption $f>0$. The proof is not detailed here since it follows the same lines as the last paragraph of the proof of Lemma \ref{lem:Ar}, decomposing $\bS f \bS=g^+ -g^-$ with $g^+,g^-\geq 0$ in $\Ktr$ such that $|\bS f\bS |=g^+ + g^-$.
%

Let us thus prove the uniqueness of $\rho$. By linearity this amounts to proving that if $\rho$ in the domain of $\fA$ solves $\rho+ \lambda \fA(\rho) =0$ then $\rho=0$. Take any such $\rho$; by assumption $\bL\rho\bL^\dag \in\Ktr$. While for $r=1$ we must still prove uniqueness, for $r\in[0,1)$, according to Lemma~\ref{lem:Ar}, there exists a \emph{unique} solution $\rho_r \in \Ktr$ to equation~\eqref{eq:Ar} where $f$ is replaced by  $\tilde{f}=(1-r)\lambda \bL\rho\bL^\dag \in\Ktr$, {i.e.~satisfying
$$
\tfrac{\bI+\lambda \bL^\dag\bL}{2} \rho_r + \rho_r \tfrac{\bI+\lambda \bL^\dag\bL}{2} = \lambda \bL \rho \bL^\dag +
r \lambda \bL (\rho_r-\rho) \bL^\dag \; .
$$
Since $\rho \in \Ktr$ satisfies this equation, we must have $\rho_r=\rho$.}
Moreover we have
$$
\tr{|\rho|} = \tr{|\rho_r|} \leq \tr{|\tilde{f}|} = (1-r)\lambda \tr{|\bL\rho\bL^\dag|},
$$
with $\tr{|\bL\rho\bL^\dag|}$ {\as independent of $r$, as $\rho$ was selected before introducing the $r$-related problem, and finite by Lemma \ref{lem:NormL}. For any fixed $\rho \in \KL$ that solves $\rho+ \lambda \fA(\rho) =0$, denoting $x = \tr{|\rho|}$ and $y=\tr{|\bL\rho\bL^\dag|}$ which are both finite, we thus have the standalone inequality
$$ 0 \leq x \leq  (1-r)\lambda y \quad \text{for all } r \in [0,1) \, .$$
This implies  $x = 0$ and we conclude that $\rho=0$.}

{\as Lemma~\ref{lem:ConvRhor} ensures that, for any $\lambda \geq 0$, the solution $\rho$ of $(\bI + \lambda \fA)(\rho) = f$ is nonnegative when $f$ is nonnegative. Since  $\rho(t) = \lim_{n\mapsto +\infty}\left( \left(\bI + \frac{t}{n}\fA\right)^{-1}\right)^{n}(\rho_0)$ for all $t\geq 0$ (see e.g.~\cite{BrezisBook1987})), $\rho(t)$ is nonnegative as soon as $\rho_0$ is also nonnegative.}
\end{proof}


\section{Asymptotic Convergence} \label{sec:AC}

We now restrict our attention to the space $\cD$ of quantum states, i.e.~nonnegative operators $\rho$ of trace one.
This is justified as by construction, the quantum dynamics \eqref{eq:mod} preserves the trace and the positivity of $\rho_0$. We characterize convergence in two ways. First, we prove that every trajectory of \eqref{eq:mod} in $\KL$ converges to a unique equilibrium point $\bar\rho$ whose support is spanned by the coherent states $\{ \q{\alpha_m} \}_{m=1,2,...,k}$ with $\alpha_m = \alpha \, e^{2i\pi m/k} \in \mathbb{C}$. Second, we identify invariants of the dynamics, which allow to readily determine to which $\bar\rho$ any particular initial state $\rho_0$ would converge.

\subsection{Unique limit point}

The first step towards proving asymptotic convergence is to identify an efficient Lyapunov function.
From Theorem~\ref{thm:Cauchy} we know that $\Lnorm{\rho_t}$ is non-increasing.  We will now characterize convergence with the following Lyapunov function
$$V(\rho) = \Lnorm{\rho} - 1\, .$$


\begin{lem}\label{lem:Lyap}
For any quantum state $\rho\geq 0 $ in $\KL$ and belonging also to the domain of the super-operator $\fA$ of Theorem~\ref{thm:Cauchy},  we have
$$
V(\rho) = \tr{\bL \rho \bL^\dag} \geq 0, \quad -\tr{\bL \fA(\rho) \bL^\dag} \leq -k!\, V(\rho) \;\; \text{and}\quad
{\as \tr{\bS \fA(\rho) \bS} = \tr{\sqrt{\bM} \bL\rho \bL^\dag \sqrt{\bM} }}
$$
where the positive operator $\bM$ is given by~\eqref{eq:M}.
\end{lem}
\begin{proof}{\as  By density of the density  operators  belonging to  $\Kf$ into the set of density operators $\rho$  belonging to the domain of $\fA$ equipped with the norm $\Lnorm{\bullet}+ \Lnorm{\fA(\bullet)}$, it  is enough to prove these three  statements with density operators $\rho\in\Kf$. Thus   we can perform all the computations as in the finite dimensional case where   $\tr{A B}=\tr{BA}$ for any operators $A$ and $B$.}
For the first statement, noting that a quantum state $\rho$ is positive and of trace one, we get:
\begin{eqnarray*}
\tr{\bL \rho \bL^\dag}+1 = \tr{\bL \rho \bL^\dag + \rho}= \tr{(\bI+\bL^\dag\bL)\rho}= \tr{\sqrt{\bI+\bL^\dag\bL}~\rho~\sqrt{\bI+\bL^\dag\bL}}= \Lnorm{\rho}
.
\end{eqnarray*}
For the second statement,  simple computations yield
\begin{multline*}
-\tr{\bL \fA(\rho) \bL^\dag} = \tr{\bL \rho  \bL^\dag [\bL^\dag,\bL] }
= - \tr{\bL \rho  \bL^\dag \bM  }=- \tr{\sqrt{\bL \rho  \bL^\dag} \bM\sqrt{\bL \rho  \bL^\dag}  } \\
\leq  - k!\, \tr{\sqrt{\bL \rho  \bL^\dag}\bI \sqrt{\bL \rho  \bL^\dag}} = -k!\,V(\rho) \, .
\end{multline*}
Here we have used $\bM$ from~\eqref{eq:M} and the rough estimate $\bM \geq  (\bN+\bI) k! \geq k! \bI$.
{\as For the third statement we have
\begin{multline*}
\tr{\bS\fA(\rho) \bS} = \tfrac{1}{2}\tr{\sqrt{\bI +  \bL^\dag\bL}\big( \bL^\dag \bL \rho + \rho \bL^\dag \bL\big)\sqrt{\bI +  \bL^\dag\bL}}
- \tr{\sqrt{\bI +  \bL^\dag\bL}\bL \rho \bL^\dag \sqrt{\bI +  \bL^\dag\bL}}
\\
= \tr{(\bI +  \bL^\dag\bL) \bL^\dag \bL \rho }
- \tr{\bL^\dag (\bI +  \bL^\dag\bL)\bL \rho }= \tr{\bL^\dag (\bL\bL^\dag-\bL^\dag\bL)\bL\rho}
\end{multline*}
where $\bL\bL^\dag-\bL^\dag\bL=\bM$.}
\end{proof}

This means that $V(\rho)$ is a  Lyapunov function ensuring exponential converge towards a steady-state $\bar\rho$ depending on the initial condition in the following sense.


\begin{thm} \label{thm:convergence}
Consider the unique trajectory  $[0,+\infty[\ni t\mapsto \rho(t)\in\KL$ solution of~\eqref{eq:mod} with initial condition $\rho(0)=\rho_0$ non-negative,  of trace one and  in the domain of $\fA$. Then there exists  $\bar\rho \in \KL$ nonnegative and of trace one, with support in $$\cH_{\alpha,k} = \text{span}\bigg\{\q{\alpha_m}: \alpha_m = \alpha \, e^{2i\pi m/k}\;, \;\; m=1,2,...,k\bigg\},$$ such that $\rho$ converges to $\bar\rho$ in $\KL$. Moreover, we have exponential convergence towards $\cH_{\alpha,k}$ in the sense:
$$ \tr{\big| \bL (\rho(t)-\bar\rho) \bL^\dag \big|} \leq \tr{\bL |\rho_0 -\bar\rho| \bL^\dag} \; e^{-k!\, t} \; .$$
\end{thm}
\begin{proof}
{\as For any given $\rho_0$ is in the domain of $\fA$, from Theorem~\ref{thm:Cauchy},  we know that $\rho(t)$ is also in the domain of $\fA$ for all $t\geq 0$, is nonnegative and of trace one}. According to Lemma \ref{lem:Lyap} the function $f(t) = V(\rho(t))$ converges exponentially to zero as $t$ increases.

{\as From Lemma~\ref{lem:Lyap} we have
$$
0\leq \tr{\sqrt{\bM} \bL \rho(t) \bL^\dag \sqrt{\bM}} =  \tr{\bS \fA(\rho(t))\bS} \leq \tr{|\bS \fA(\rho(t))\bS|}=\Lnorm{\fA(\rho(t))} \leq \Lnorm{\fA(\rho_0)}
$$
since according to Theorem~\ref{thm:Cauchy} the norm in $\KL$ of $\fA(\rho(t))$ is time-decreasing. This implies  that $\tr{\sqrt{\bM}\bL \rho(t) \bL^\dag \sqrt{\bM}} $ is bounded in time. Denote by $\mathcal{K}_{\sqrt{M}L}(\mathcal{H})$ the  Banach space of  operators $\xi $  equipped with the   trace norm $\tr{\left|\sqrt{\bI + \bL^\dag \bM \bL}~\xi ~ \sqrt{\bI +\bL^\dag \bM \bL}\right|}$. Since $\bM\geq k! (\bN+\bI)$ is unbounded, it is not difficult to prove that the injection of $\mathcal{K}_{\sqrt{M}L}(\mathcal{H})$ in $\KL$ is compact. But  $\rho_0$ belongs to the domain of $\fA$. Thus  the above inequality implies, since $\rho(t)\geq 0$,
 that $\rho(t)$ belongs to  $\mathcal{K}_{\sqrt{M}L}(\mathcal{H})$  and is bounded uniformly versus $t$.}
Thus the trajectory $\{\rho(t)~|~t\geq 0\}$ is precompact in $\KL$, which means that it must have an adherent point $\bar\rho\in \KL$ for $t$ tending towards infinity.

Lemma~\ref{lem:Lyap} implies $\tr{\bL \bar\rho \bL^\dag}=0$, i.e.~$\bar\rho$ is a steady state and its support is contained in the kernel of $\bL$. The latter is spanned by vectors satisfying $\ba^k \q{\psi} = \alpha^k \q{\psi}$ and writing $\q{\psi} = \sum_{m \in \mathbb{N}} \, \psi_m \q{m}$ we get the recurrence relation
$$ \psi_{m+k} = \frac{\alpha^k}{\sqrt{(m+k)...(m+2)(m+1)}} \; \psi_m \quad \text{for all } m\geq 0.$$
This leaves $k$ degrees of freedom to initialize the recurrence(s), which besides that is satisfied by the $k$ states mentioned in the statement as span of $\cH_{\alpha,k}$. Thus the kernel of $\bL$ indeed coincides with $\cH_{\alpha,k}$.

For the first part, there remains to show that the adherence point is unique. {\as  Since $\rho_t-\bar\rho$ is the solution of~\eqref{eq:mod} with initial condition $\rho_0-\bar\rho$ in the domain of $\fA$, by Theorem~\ref{thm:Cauchy} its norm $\Lnorm{\rho(t)-\bar\rho}$ is non-increasing and thus $\rho(t)$ converges towards $\bar\rho$ in $\KL$  that is unique.}

For exponential convergence,   decompose $\rho(0)-\bar\rho= \zeta^+_0 -\zeta^-_0$ (positive and negative part) where the evolution of $\zeta^+_t$ and $\zeta^-_t$ are defined by
\begin{align*}
\frac{d}{dt}\zeta^+ =-\fA(\zeta^+), \quad \zeta^+(0) = \zeta^+_0 \\
\frac{d}{dt}\zeta^- = -\fA(\zeta^-), \quad \zeta^-(0) = \zeta^-_0
\end{align*}therefore, by linearity, $\rho(t)-\bar{\rho}= \zeta_t^+ - \zeta_t^-$ and we get
\begin{eqnarray*}
\tr{| \bL (\zeta_t^+ - \zeta_t^-) \bL^\dag |} & \leq & \tr{|\bL \zeta_t^+ \bL^\dag |} + \tr{|\bL \zeta_t^- \bL^\dag|} \\
& = & \tr{\bL (\zeta_t^+ + \zeta_t^-) \bL^\dag } \leq \tr{\bL |\rho_0 - \bar{\rho}| \bL^\dag} \; e^{-k!\, t} \; ,
\end{eqnarray*}
where the last inequality follows from Lemma \ref{lem:Lyap} and $\zeta_0^++\zeta_0^- = |\rho(0) - \bar{\rho}|$.
\end{proof}

Regarding exponential convergence, since the kernel of $\bL$ coincides with $\cH_{\alpha,k}$ and the spectrum of $\sqrt{\bL^\dag \bL}$ has no accumulation point at 0, the role of the $\bL$ operators is essentially to project $\rho(t)$ onto the complement of $\cH_{\alpha,k}$. I.e.~if $\mathcal{P}_{k}$ is the orthonormal projector from $\cH$ onto $\cH_{\alpha,k}$, then there exists some $c>0$ and $c'>0$ (depending on $\rho_0$) such that
$$ \tr{| (\bI-\mathcal{P}_{k}) (\rho(t)-\bar\rho) (\bI-\mathcal{P}_{k}) |} \leq \; c \; \tr{\bL |\rho_0 -\bar\rho| \bL^\dag} \; e^{-k!\, t} \; \leq \; c'  \; e^{-k!\, t} \; .$$
This expresses exponential convergence towards $\cH_{\alpha,k}$ in the trace-norm. The inclusion of the unbounded operator $\bL$
on the left-hand side
makes the statement of Theorem \ref{thm:convergence} slightly stronger.

Note that the Hermitian operators with support on $\cH_{\alpha,k}$ belong to $\KL$. They form a $k^2$-dimensional real subspace of $\KL$, since $\cH_{\alpha,k}$ is a $k$-dimensional complex subspace of $\cH$. Thus the exponential convergence neglects $k^2$ directions in $\KL$. The behavior of the system in these directions, i.e.~how the limit $\bar\rho$ depends on $\rho_0$, is clarified in the following section.


\subsection{Invariants of the dynamics} \label{sec:INV}

We will derive invariants of the dynamics in terms of bounded  Hermitian operators on $\cH$, i.e., in terms of physical observables. For this, we first define a continuous extension to any $\rho_0 \in \Ktr$ of the superoperator mapping $\rho_0$ to its corresponding limit $\bar\rho$ according to Theorem \ref{thm:convergence}.

Consider the   map $\fK$  from $\Kf$ to $\Ktr$ defined as follows: for any $\xi\in\Kf$, $\fK(\xi)=\lim_{t\mapsto +\infty} \rho(t)$ where $\rho(t)$ is the unique solution of~\eqref{eq:mod} with $\rho_0=\xi$. By construction, $\fK$ is linear, trace preserving, completely positive and a contraction for the nuclear norm \cite[Chapter 9]{NielsenChuang}:
$$
\forall \xi_1,\xi_2\in\Kf, \quad \tr{\big|\fK(\xi_1)-\fK(\xi_2)\big|} \leq \tr{|\xi_1-\xi_2|}
.
$$ Since $\Kf$ is dense in $\Ktr$, we can extend the domain of definition of $\fK$ to all $\xi \in \Ktr$ by continuity. From there, we get the following result.

\begin{thm}\label{thm:invariants}
There exist $k^2$ linearly independent Hermitian bounded operators $\bzq_{m,m'}$,  $m,m'=1,2,...,k$, which are invariant under the dynamics \eqref{eq:mod}, i.e.~for which
$$ \tr{\bzq_{m,m'} \, \rho_t } = \tr{\bzq_{m,m'} \, \rho_0 }$$
for any trajectory $[0,+\infty) \ni t \mapsto \rho_t \in \KL$.

Moreover, the linear space of invariant  Hermitian operators spanned by $\{ \bzq_{m,m'} \}_{m,m'=1...k}$ contains in particular the $k$ operators
\begin{eqnarray*}
\boldsymbol{Q}^{\cos}_{m} = \sum_{n \in \mathbb{N}} \cos\left(\frac{2\pi mn}{k}\right) \ket{n}\bra{n} \quad &\text{ for } & m=0,1,...,\lceil \tfrac{k-1}{2} \rceil \, ;\\
\boldsymbol{Q}^{\sin}_{m} = \sum_{n \in \mathbb{N}} \sin\left(\frac{2\pi mn}{k}\right) \ket{n}\bra{n} \quad &\text{ for } & m=1,...,\lfloor \tfrac{k-1}{2} \rfloor \, .
\end{eqnarray*}
\end{thm}
\begin{proof}
The image of $\fK$ has support on $\cH_{\alpha,k}$, hence we can write
\begin{eqnarray*}
\fK(\rho_0) &=& \sum_{m=1}^k \sum_{m'=1}^{m-1}  Q_{m,m'}(\rho_0)\, (\ket{\alpha_m}\bra{\alpha_{m'}}+\ket{\alpha_{m'}}\bra{\alpha_m}) \\
&& + \sum_{m=1}^k  Q_{m,m}(\rho_0)\, \ket{\alpha_m}\bra{\alpha_m}\\
&& + \sum_{m=1}^k \sum_{m'=m+1}^{k}  Q_{m,m'}(\rho_0)\, (i\ket{\alpha_m}\bra{\alpha_{m'}}-i\ket{\alpha_{m'}}\bra{\alpha_m})
\end{eqnarray*}
with $Q_{m,m'}(\rho_0)$ real. Since $\fK$ is a linear continuous map from $\Ktr$ to $\Ktr$, the  $Q_{m,m'}: \Ktr\ni \rho_0\mapsto Q_{m,m'}(\rho_0)$ are linear continuous maps from $\Ktr$ to $\RR$. Thus  they belong to the dual of $\Ktr$ corresponding to the set of   bounded operators on $\cH$:
$$Q_{m,m'}(\rho) \; = \; \tr{\bzq_{m,m'}\, \rho}$$
where  $\bzq_{m,m'}$ are \emph{bounded} Hermitian operators.
Since any Hermitian operator $\bar\rho$ with support on $\cH_{\alpha,k}$ belongs to the image of $\fK$, and this is a $k^2$-dimensional set, the $k^2$ operators $\{ \bzq_{m,m'} \}$ must indeed be linearly independent. Moreover, since all $\rho_t$ along the trajectory defined by~\eqref{eq:mod} starting at $\rho_0 \in \KL$ are associated to the same $\bar\rho=\fK(\rho_t)$, we have that $\tr{\bzq_{m,m'}\, \rho_t}$ is constant along such trajectories.

The particular operators $\boldsymbol{Q}^{\cos}_{m}$ and $\boldsymbol{Q}^{\sin}_{m}$ are indeed bounded and e.g.
$$\tr{ \boldsymbol{Q}^{\cos}_{m} \xi } = \sum_{n\in\mathbb{N}}\,  \cos\left(\frac{2\pi mn}{k}\right)\, \xi_n $$
for any $\xi \in \Ktr$, where we have written $\xi_n = \bra{n} \xi \ket{n}$; let also $\xi_{\ell,n} = \bra{\ell} \xi \ket{n}$. Computing $\mathfrak{L}_{\bL}(\rho)$ explicitly in the canonical basis, we get\vspace{2mm}

$\sum_{n \in \mathbb{N}}\, \cos\left(\frac{2\pi mn}{k}\right)\, \tfrac{d}{dt}\xi_n$
\vspace{-2mm}
\begin{eqnarray*}
& = & \sum_{n \in \mathbb{N}}\; \xi_{n+k} \, (n+k)...(n+1) \; \left( \cos\left(\tfrac{2\pi mn}{k}\right) - \cos\left(\tfrac{2\pi m(n+k)}{k}\right) \right) \\
& & - \frac{\alpha^k}{2} \sum_{n \in \mathbb{N}}\; (\xi_{n,n+k}+\xi_{n+k,n}) \; \sqrt{(n+k)...(n+1)} \; \left( \cos\left(\tfrac{2\pi mn}{k}\right) - \cos\left(\tfrac{2\pi m(n+k)}{k}\right) \right)\;\; = 0 \, ,
\end{eqnarray*}
which indeed proves invariance of $\boldsymbol{Q}^{\cos}_{m}$ along trajectories; for $\boldsymbol{Q}^{\sin}_{m}$ the proof is exactly the same.
\end{proof}

Note that the pure states on $\cH_{\alpha,k}$ span a space of dimension $k$, thus in this sense having $k$ invariants can provide significant insight about the link between $\rho_0$ and its associated $\bar\rho$. In the present case, this link is particularly meaningful.

We have $\boldsymbol{Q}^{\cos}_{0}= \bI$, so this particular invariant just expresses conservation of the trace of $\rho$. The other particular invariants feature as eigenstates so-called ``Schr\"odinger cat states'' of the harmonic oscillator, more soberly called coherent quantum superpositions of mesoscopic states, and whose general form is defined in equation \eqref{eq:defcat}. More precisely, denoting
$$\ket{C^\ell_\alpha} = \frac{1}{\vartheta} \sum_{m=1}^k \, e^{2 i \pi \ell m / k} \, \ket{\alpha_m}$$
with $\vartheta$ a normalizing constant, a few computations based on the definitions directly yield:
$$\boldsymbol{Q}^{\cos}_{m} \, \ket{C^\ell_\alpha} \; = \cos\left(\tfrac{2\pi\ell m}{k}\right) \, \ket{C^\ell_\alpha}
\quad \text{and} \quad \boldsymbol{Q}^{\sin}_{m} \, \ket{C^\ell_\alpha} \; = \sin\left(\tfrac{2\pi\ell m}{k}\right) \, \ket{C^\ell_\alpha} \; .$$
The Schr\"odinger cats $\ket{C^\ell_\alpha}$ are specifically quantum states with no classical analogue, and they are a promising tool towards implementing quantum IT applications \cite{MirrahimiCatComp2014}, thanks to their inherent insensitivity to part of the typical perturbations present in quantum systems. The known invariants allow us to predict towards which fraction of each cat $\ket{C^\ell_\alpha}$ an arbitrary initial state $\rho_0$ will evolve. This can be useful for investigating more precisely the sensitivity of information encoded in such cat states to typical perturbations.

The authors of \cite{MirrahimiCatComp2014} discuss the case $k=2$ in this direction. In that case, the $k$ particular invariants of Theorem \ref{thm:invariants} reduce to the identity and the parity operator $\sum_{n \in \mathbb{N}} \, (-1)^n \, \ket{n}\bra{n}$. Explicit expressions for the $k^2-k=2$ remaining linearly independent invariants are also provided in terms of Bessel functions. Generalizing these expressions for $k>2$ remains for future work.


\section{Conclusion}

We have proved, with a rigorous infinite-dimensional treatment, that a harmonic oscillator governed by a Lindblad master equation where the typical drive and loss channels are $k$-photon processes instead of single-photon ones, converges to a protected subspace spanned by $k$ coherent states of the same amplitude and uniformly distributed phases. We have also proved the existence of $k^2$  invariant bounded observables (Hermitian operators), i.e.~whose expectation value is conserved by the dynamics. Knowing these invariants would allow to directly predict the final state $\bar \rho$ towards which a given $\rho_0$ converges. We have provided explicit expressions for $k$ such invariant observables, whose eigenstates are Schr\"odinger cat states belonging to the protected subspace and which appear as robust candidates to encode quantum information.

The infinite-dimensional arguments use the Hille-Yosida theorem and a Lyapunov analysis in a particular family of Banach spaces $\KL$ whose metric is built directly from the $k$-photon Lindblad operator. For practical purposes, our contribution is to show that there exists a self-consistent way to indeed have infinite-dimensional convergence to a so-called protected subspace in this model, and to prove exponential  convergence speed.

We guess that the methods used here to study  the well-posedness of  our particular infinite-dimensional Lindbald master equation can be adapted to other models appearing in reservoir engineering for  cavity and/or  circuit quantum electrodynamics.


\section*{Acknowledgment}

The authors thank Mazyar Mirrahimi and Joachim Cohen for many useful discussions.

\bibliographystyle{plain}

\section*{Appendix}

\begin{thm}[Hille-Yosida in Banach spaces \protect{\cite[Chap.7]{BrezisBook1987}}]\label{thm:HillYosida}
Let $E$ be a Banach space and $A$ an m-accretive operator on $E$, i.e., the domain $D(A)$ of $A$ is dense in $E$ and for every $\lambda >0$, $I+\lambda A$ is a bijection from $D(A)$ into $E$ with, for all $u\in E$, $\left\| (I+\lambda A)^{-1} u \right\| \leq \|u\|$. Then for any $u_0\in D(A)$, there exists a unique function
$[0,+\infty[\ni t \mapsto u(t) \in D(A)$ that is continuously differentiable such that
$$
\dotex u  + A(u)=0 \text{ for } t\in [0,+\infty[, \quad u(0)=u_0.
$$
Moreover $\forall t \geq 0$, $ \| u(t)\|\leq \|u_0\|$ and $\|\dotex u(t)\| = \|A(u(t))\| \leq  \|A(u_0)\|$
\end{thm}

\end{document}